\documentclass[12pt]{iopart}
\usepackage{amsfonts,amssymb,amsthm}
\usepackage{enumerate}
\usepackage{color}
\usepackage[english]{babel}

\newtheorem*{theorem*}{Theorem}
\newtheoremstyle{ddtheoremstyle} % name
    {\topsep}                    % Space above
    {\topsep}                    % Space below
    {}                           % Body font
    {}                           % Indent amount
    {\scshape}                   % Theorem head font
    {.}                          % Punctuation after theorem head
    {.5em}                       % Space after theorem head
    {}  % Theorem head spec (can be left empty, meaning ‘normal’
\newtheoremstyle{ddstyle} % name
    {\topsep}                    % Space above
    {\topsep}                    % Space below
    {}                           % Body font
    {}                           % Indent amount
    {\scshape}                   % Theorem head font
    {.}                          % Punctuation after theorem head
    {.5em}                       % Space after theorem head
    {}  % Theorem head spec (can be left empty, meaning ‘normal’)
\let\vphi\varphi

\let\p\partial

\let\ds\displaystyle

\let\noi\noindent
\theoremstyle{ddstyle}
\newtheorem*{corollary}{Corollary}
\newtheorem*{remark}{Remark}

\newtheorem{manualtheoreminner}{Example}
\newenvironment{Example}[1]{%
  \IfBlankTF{#1}
    {\renewcommand{\themanualtheoreminner}{\unskip}}
    {\renewcommand\themanualtheoreminner{#1}}%
  \manualtheoreminner
}{\endmanualtheoreminner}

\newtheorem*{Proposition}{Proposition}
\newtheorem*{lemma}{Lemma}

\newcounter{theorem}
\begin{document}
\title{Characteristic integrals and general solution of the Ferapontov-Shabat-Yamilov lattice} % 
\author{Dmitry K. Demskoi}
\address{School of Computing and Mathematics,
Charles Sturt University, NSW 2678, Australia}
\eads{\mailto{ddemskoy@csu.edu.au }}

\date{}
%\maketitle
\begin{abstract}
For the finite (non-periodic) systems obtained from a lattice introduced by Ferapontov and  independently by Shabat and Yamilov,
we present a quadrature-free general solution and a recurrent formula for the characteristic 
integrals.
The derivation of these formulae relies on the underlying determinantal equations.
We illustrate the results using a two-component system.
\end{abstract}
\section{Introduction}
The lattice equation
\begin{equation}
	\p^2_{tx}v_n=\p_t v_{n} \p_x v_{n}\left(\frac{1}{v_n-v_{n-1}}-\frac{1}{v_{n+1}-v_n}\right),
	\label{ShY}
\end{equation}
 which we refer to as Ferapontov-Shabat-Yamilov (FSY) lattice, was introduced independently 
 in  \cite{yamil} and \cite{ferapont}. In the work of Shabat and Yamilov, equation (\ref{ShY})
was derived as one of the examples in the list of two-dimensional generalisations of the Toda-type lattices.
There is a close connection between the latter and the nonlinear Shr\"odinger-type (NLS) systems:
the lattices serve as auto-B\"acklund transformations for the NLS-type systems. 
Respectively, the two-dimensional generalisations of the Toda-type lattices play a role of 
auto-B\"acklund transformations of Davey-Stewartson-type systems.
It turns out that (\ref{ShY}) is compatible with the well-known Ishimori equation
which is a two-dimensional generalization of the Heisenberg spin model \cite{ishimori,garhabib}.

In the work of Ferapontov, the lattice (\ref{ShY}) emerged in the context of 
two-component hydrodynamic-type systems in Riemann invariants. 
It was shown that the characteristic velocities of such systems, 
related by the Laplace transformation, satisfy lattice (\ref{ShY}).
Ferapontov also investigated periodic reductions of the corresponding chain of the Laplace transformation
and lattice (\ref{ShY}).

In regards to the integrability of reductions of lattice equations, 
normally the periodic boundary conditions, i.e. \(v_{n+N}=v_n\), lead to solitonic equations, whereas finite 
(non-periodic) reductions, often referred to as {\it cutoff constraints}, yield linearisable equations.
For hyperbolic lattices like (\ref{ShY}), the finite reductions turn out to be Darboux
integrable. This type of integrability is characterised by existence of complete sets
of integrals along each characteristic direction.

A finite reduction of the lattice (\ref{ShY}) corresponding to the boundary condition 
\begin{equation}
	v_0=0,\ \  v_{N}=\infty
	\label{myboundary0N}
\end{equation}
was introduced by the author \cite{demint} in connection with the determinantal equation 
\begin{equation}
	\left|\begin{array}{cccc}
		u & u_t & \ldots & u_{t \ldots t} \\
		u_x & u_{t x} & \ldots & u_{t \ldots t x} \\
		\vdots & \vdots & \ddots & \vdots \\
		u_{x \ldots x} & u_{t x \ldots x} & \ldots& u_{t \ldots t x \ldots x}
		\end{array}\right|=0,
		\label{wn0}
\end{equation}
where
\[
u_{t \ldots t}=\partial_t^{N-1} u,\, u_{x \ldots x}=\partial_x^{N-1} u,\, u_{t \ldots t x \ldots x}=\partial_x^{N-1} \partial_t^{N-1} u,\,\mbox{etc.}
\]	
It was pointed out that the quantities below (see also formula (\ref{vdef}))
\begin{equation}
	v_1=\frac{1}{u},\, v_2=\frac{u_{tx}}{\left|\begin{array}{ll}
		u & u_t \\
		u_x & u_{t x}
		\end{array}\right|},\,
		v_3=\frac{\left|\begin{array}{ll}
			u_{tx} & u_{ttx} \\
			u_{txx} & u_{tt xt}
			\end{array}\right|}{
				\left|\begin{array}{cccc}
					u & u_t & u_{t t} \\
					u_x & u_{t x} &  u_{ttx} \\
					u_{xx} & u_{txx} & u_{ttxx}
					\end{array}\right|
			},\dots
	\label{newvarsv}
\end{equation}
satisfy lattice (\ref{ShY}).

On the other hand, equation (\ref{wn0}) is known \cite{darb,leznov,tsarev} to be closely related to the two-dimesional Toda lattice
in the algebraic form
\begin{equation}
	w_0=1,\ \ \ds \p^2_{tx} \ln w_n=\frac{w_{n+1}w_{n-1}}{w_n^2},\ \ w_N=0,\ \ 1<n< N.
	\label{todaln}
\end{equation}
Specifically, the relationship between the scalar equation (\ref{wn0}) and the lattice (\ref{todaln}) is realised if we introduce the lattice variables
$w_n$ according to the formula
\begin{equation}
	\begin{array}{l}
	 w_1=u,\,  w_2=\left|\begin{array}{ll}
		u & u_t \\
		u_x & u_{t x}
		\end{array}\right|,\,
	 w_3=\left|\begin{array}{cccc}
			u & u_t & u_{t t} \\
			u_x & u_{t x} &  u_{ttx} \\
			u_{xx} & u_{txx} & u_{ttxx}
			\end{array}\right|,\dots %w_N=0,
		\end{array}
	\label{wndefexpl}
\end{equation}
i.e. quantities $w_n$ represent the principal minors of the $n\times n$ matrix
corresponding to the determinant in the left-hand side of formula (\ref{wn0}).

The boundary condition
	\begin{equation}
		v_0=a,\, v_N=b,
		\label{habibconstr}
	\end{equation}
	where \(a,b=\mbox{const},\) is equivalent to 
(\ref{myboundary0N}) 
due to invariance of the lattice (\ref{ShY}) under the M\"obius transformation
	\begin{equation}
		\vphi_n=\frac{\alpha v_n+\beta}{\gamma v_n+\delta}.
		\label{mobius}
	\end{equation}
Thus, by choosing the parameters such that (\ref{mobius}) becomes
	\begin{equation}
		\vphi_n=\frac{(a-b) v_n}{a-b-v_n}+a,\ \ a\ne b,	
		\label{mobiusab}
	\end{equation}
	we observe that 
\begin{equation}
	\vphi_0=a,\ \ \vphi_N=b,
	\label{habibbound}
\end{equation}
provided that $v_n$ satisfies the condition (\ref{myboundary0N}). 
Furthermore, the degenerate case $b=a$ can be recovered in two steps: 
\begin{enumerate}
	\item 
Apply the reduction $v_{N-1}=0$ in the system obtained from the condition (\ref{myboundary0N}). 
This implies the homogeneous boundary condition $v_0=0,\, v_{N-1}=0$;
\item Shift the variables: $\vphi_n= v_n+a.$
\end{enumerate}

The Darboux integrability of the Toda lattice (\ref{todaln}), defined as the existence of complete sets of characteristic integrals, is well-established \cite{shabyam}. Consequently, the lattice (\ref{ShY}) with the boundary condition (\ref{myboundary0N}), or equivalently (\ref{habibconstr}), is also Darboux integrable. 
Furthermore, an explicit formula for the characteristic integrals was provided in \cite{demint} (see formula (10) therein), where the integrals of (\ref{ShY}) were expressed in terms of a single variable, \( u = 1/v_1 \), and its derivatives. 
This formula yields characteristic integrals for all lattices derived from the scalar equation (\ref{wn0}) by introducing lattice variables.

However, a more natural representation of the integrals should involve all lattice variables and avoid mixed derivatives. While such expressions can be derived from the aforementioned formula, the process becomes increasingly computationally expensive as the number of lattice variables grows. The related approaches in \cite{smirnov} and \cite{habibsak}, which rely on the system’s linear representation (Lax pair), allow the construction of integrals as coefficients of a factorised differential operator. The drawback of this approach, however, is its labour-intensive nature, as it requires expanding the differential operator to determine its coefficients.  

Thus, it would be advantageous to find a direct formula for the integrals that avoids the need to compute additional objects.  
In the following section, equation (\ref{ShY}) will be derived from the determinantal equation (\ref{wn0}) \textit{ab initio}.  
Subsequently, recurrent formulae for the integrals and the general solutions will be calculated in the variables of the hyperbolic systems obtained by constraining 
lattice (\ref{ShY}) with boundary conditions (\ref{myboundary0N}) and (\ref{habibbound}).  
In conclusion, we will discuss discrete and semi-discrete analogues of (\ref{ShY}) that are related 
to the determinantal equation (\ref{wn0}).  

\section{Determinantal structure of the FSY lattice} 
First, we introduce some notation to prepare the ground for a derivation of the lattice (\ref{ShY}).
Consider a function $w_n$ defined by the $n$-th order determinant
\begin{equation}
	w_n(u)=\mbox{det}\left(a_{ij}\right),
	\label{wndef}
\end{equation}
where 
\[\quad a_{i+1,j+1}=\frac{\partial^{i+j} u}{\partial x^i \partial t^j},\ \ i,j=1\dots n-1.\] 
Hence, the argument of $w_n$ is the entry in the top left corner of the determinant.
Unless otherwise stated, it will be suppressed, i.e. $w_n=w_n(u).$ 
In this notation, equation (\ref{wn0}) can be re-stated as follows
\begin{equation}
	w_N(u)=0.
	\label{wN01line}
\end{equation}
Without any further restrictions, $w_n$ satisfies the lattice equation 
	\begin{equation}
		w_0=1,\ \ \ds \p^2_{tx} \ln w_n=\frac{w_{n+1}w_{n-1}}{w_n^2}
		\label{todalninf}
	\end{equation}
which is infinite in the positive direction.
Formula (\ref{wndef}) is essentially a definition of new 
variables $w_n$ which transform the scalar equation (\ref{wn0}) into a chain of hyperbolic
equations (\ref{todalninf}).
Imposing the condition $w_N(u)=\mbox{const}$ turns (\ref{todalninf})
into a Darboux integrable system of PDEs. 
However, we note that (\ref{todalninf}) then implies $w_{N+1}(u)=0,$ thus, without loss of generality we may assume 
that equation (\ref{wN01line}) is satisfied.

\subsubsection*{Jacobi determinantal identity.}

We denote the minor of the entry in the $p$-th row and $q$-th column of the matrix $(a_{ij})$ as $m_{pq}$ 
and the minor obtained from $(a_{ij})$ by removing its $p$- and $q$-th rows as well as $r$- and $s$-th columns as $m_{pqrs}$. 
In such notation, the Jacobi determinant identity can be written as follows \cite{hirotajac}
\begin{equation}
	m_{pqrs} w_n=m_{pr}m_{qs}-m_{ps}m_{qr}. 
	\label{sylvester}
\end{equation}
When it is necessary to indicate the order of minors $m_{pq}$ and $m_{pqrs}$ explicitly, we write $m_{n;pq}$ and $m_{n;pqrs}$ respectively.
This notation means that the minors are obtained by deleting rows and columns from the $n\times n$ matrix $(a_{ij}).$
We will refer to instances of the Jacobi identity by indicating the list of indices $(p,q,r,s)$.

Consider the ratios of the determinants:  
\begin{equation}
J_{n ; p}=\frac{m_{p n}}{w_{n-1}},\ \ I_{n;p}=\frac{m_{np}}{w_{n-1}},\ \  p=1, \ldots, n
\label{defIJ}
\end{equation}
with the values for \(p = n\) and \(p = 0\):  
\[J_{n;n}=I_{n;n}=1,\ \ J_{n;0}=I_{n;0}=0.\]
The quantities \(J_{n;p}\) and \(I_{n;p}\) are connected by the involution \(t \leftrightarrow x\).  
Broadly speaking, they represent characteristic \(t\)- and \(x\)-integrals, respectively. A more precise description is provided below.%
\begin{Proposition}\cite{demtran}  The quantities $J_{n ; p}$ and $I_{n ; p}$ satisfy the equations
\begin{equation}
 \partial_t J_{n ; p}=\frac{w_{n-2} w_n}{w_{n-1}^2} J_{n-1 ; p},
	\label{recJ}	
\end{equation}
\begin{equation}
	J_{n ; p}=J_{n-1 ; p-1}-\partial_x J_{n-1 ; p}+J_{n-1 ; p} \partial_x \ln \frac{w_{n-1}}{w_{n-2}},
	\label{recJpair}	
\end{equation}
and 
\begin{equation}
	\partial_x I_{n ; p}=\frac{w_{n-2} w_n}{w_{n-1}^2} I_{n-1 ; p},
	\label{recI1}
\end{equation}
\begin{equation}
	I_{n ; p}=I_{n-1 ; p-1}-\partial_t I_{n-1 ; p}+I_{n-1 ; p} \partial_t \ln \frac{w_{n-1}}{w_{n-2}}.
	\label{recI2}
\end{equation}
\end{Proposition}
Now the connection between the quantities  $J_{n ; p}$ and  $I_{n ; p}$ on one hand,
and the characteristic integrals of the lattice (\ref{todaln}) on the other, is obvious.
Indeed, if $w_n=0$ for some $n=N,$ 
then, as follows from (\ref{recJ}) and (\ref{recI1}), 
formula (\ref{defIJ}) gives expressions for the characteristic integrals of equation (\ref{wn0}) as well as
of the lattice (\ref{todaln}), depending on whether $w_{n},$ present in the formulae
are interpreted as determinants or the lattice variables as in (\ref{todaln}), i.e.
\[\p_t J_{N;p}=0,\ \ \p_x I_{N;p}=0\ \ \mbox{mod}\,\,  w_N=0.\]
Thus, we have $2N-2$ integrals which equals the order
of equation (\ref{wn0}). Additionally, the independence of these integrals 
is evident from their structure in (\ref{defIJ}). 
Indeed, all these integrals have the same order, $2N-3,$ 
yet by construction each integral depends on the unique set of variables.
\begin{Example}{1.1}
	Consider the case $N=3.$ The scalar equation (\ref{wn0}) takes the form 
	\begin{equation}
		\left|\begin{array}{cccc}
			u & u_t & u_{t t} \\
			u_x & u_{t x} &  u_{ttx} \\
			u_{xx} & u_{txx} & u_{ttxx}
			\end{array}\right|=0.
		\label{wn3}
	\end{equation}
The respective lattice (\ref{todaln}) becomes a system for variables $w_1$ and $w_2:$
\begin{equation}
	w_0=1,\ \ \ds \p^2_{tx} \ln w_1=\frac{w_{2}}{w_1^2},\ \  
	\p^2_{tx} \ln w_2=0, \ \ w_3=0.
	\label{todaln3}
\end{equation}
Formula (\ref{defIJ}) gives the following expressions for the $t$-integrals
	\begin{equation}
		J_{3 ; 1}=\frac{\left|\begin{array}{ll}
			u_x & u_{t x} \\
			u_{x x} & u_{t x x}
			\end{array}\right|}{\left|\begin{array}{ll}
			u & u_t \\
			u_x & u_{t x}
			\end{array}\right|},\ \ 
			J_{3 ; 2}=\frac{\left|\begin{array}{ll}
				u & u_t \\
				u_{x x} & u_{t x x}
				\end{array}\right|}{\left|\begin{array}{ll}
				u & u_t \\
				u_x & u_{t x}
				\end{array}\right|}.
	\label{J31J32}			
	\end{equation}
	In order to express $J_{3 ; 1}$ and $J_{3 ; 2}$ in terms of $w_1$ and $w_2$ one can substitute $u=w_1$
	and eliminate mixed derivatives using (\ref{todaln3}). However, a simpler way to calculate them is to use formula (\ref{recJpair}) which gives
	\begin{equation}
		\begin{array}{ll}
			J_{2 ; 1} & \ds = J_{1 ; 0}-\partial_x J_{1 ; 1}+J_{1 ; 1} \partial_x \ln \frac{w_1}{w_0} \\
			& =\partial_x \ln w_1,
			\end{array}
			\label{J21todaln}
	\end{equation}
	Then the integrals are calculated as follows
	\begin{equation}
		\begin{array}{ll}
			 J_{3 ; 1} & \ds =J_{2 ; 0}-\partial_x J_{2 ; 1}+J_{2 ; 1} \partial_x \ln \frac{w_2}{w_1} \\[4mm]
			 & \ds =-\partial_x^2 \ln w_1+\left(\ln w_1\right)_x \partial_x \ln \frac{w_2}{w_1} \\[4mm]
			& \ds =-\frac{w_{1, x x}}{w_1}+\frac{w_{1, x} w_{2, x}}{w_1 w_2},\\[4mm]
			J_{3 ; 2} & \ds =J_{2 ; 1}-\partial_x J_{2 ; 2}+J_{2 ; 2} \partial_x \ln \frac{w_2}{w_1} \\[2mm]
			& \ds =\partial_x \ln w_1+\partial_x \ln \frac{w_2}{w_1} \\[2mm]
			& \ds =\partial_x \ln w_2 .
		\end{array}
	\end{equation}
	\end{Example}
\subsubsection*{New variables.}
Let us introduce a new variable $v_n$ as the ratio of the determinants 
\begin{equation}
	v_n=\frac{w_{n-1}\left(u_{t x}\right)}{w_n},\ \ n\ge 1,
	\label{vdef}
\end{equation} 
or the same 
\(v_n=\p_u \ln w_n.\)
It follows from this definition that
\begin{equation}
	v_1=\frac{1}{u}
	\label{v_1}	
\end{equation}
since $w_1=u.$
\begin{lemma}
Quantity $v_n$ satisfies the following equation
\begin{equation}
	\p_t v_n=-\frac{w_{n-1}\left(u_x\right) w_n\left(u_t\right)}{w_n^2}
	\label{vn_t}
\end{equation}
\end{lemma}
\begin{proof}
Identity (\ref{sylvester}) with $(p,q,r,s)=(1,n+1,1,n+1)$ can be cast into the form
\begin{equation}
w_{n-1}(u_{t x}) w_{n+1}
=w_{n}(u_{t x}) w_{n}-w_{n}(u_x) w_{n}(u_t).
	 \label{J0}
\end{equation}
Dividing it by $w_n w_{n+1}$ and employing (\ref{vdef}) we obtain
\begin{equation}
	v_{n+1}=v_{n}+\frac{w_{n}\left(u_x\right)}{w_{n}} \frac{w_{n}\left(u_t\right)}{w_{n+1} }.
	\label{JCBR}
\end{equation}
Setting $p=1$ in $(\ref{recJ})$ we  can rewrite it as follows
\begin{equation}
	\partial_t \frac{w_{n-1}\left(u_x\right)}{w_{n-1}}=\frac{w_{n-2}\left(u_x\right) w_n}{w_{n-1}^2}.
	\label{twostar}	
\end{equation}
Additionally, setting $p=1$ in $(\ref{recI2})$  and noting that $I_{n-1,0}=0$ we get the following equation
\begin{equation}
	\partial_t \frac{w_{n-1}\left(u_t\right)}{w_{n}}=-\frac{w_{n-1} w_n\left(u_t\right)}{w_n^2}.
	\label{twostar2}	
\end{equation}
Further, we employ induction on $n$. For $n=1$ equation (\ref{vn_t}) takes the form $\p_t(1/u)=-u_t/u^2$ hence satisfied.
Differentiating (\ref{JCBR}) and employing (\ref{vn_t}), (\ref{twostar}) and (\ref{twostar2}), we get
\[
\p_t v_{n+1}=-\frac{w_n\left(u_x\right) w_{n+1}\left(u_t\right)}{w_{n+1}^2}
\]
which is nothing but the upshifted formula (\ref{vn_t}).
\end{proof}
\begin{corollary}
	${}^{}$
\begin{enumerate}[i.] 
		\item 
	Due to symmetry $t\leftrightarrow x$ in the definition of $v,$ formula (\ref{vn_t}) implies the following equation
	\begin{equation}
			\p_x v_n=-\frac{w_{n-1}\left(u_t\right) w_n\left(u_x\right)}{w_n^2}.
			\label{vn_x}
		\end{equation}
\item Additionally, let us state the $x$-counterparts of formulae (\ref{twostar}) and (\ref{twostar2}), which will be useful later:
\begin{equation}
	%\begin{array}{l}
	\partial_x \frac{w_{n-1}\left(u_t\right)}{w_{n-1}} =\frac{w_{n-2}\left(u_t\right) w_n}{w_{n-1}^2}, \ \ 
	\partial_x \frac{w_{n-1}\left(u_x\right)}{w_{n}} =-\frac{w_{n-1} w_n\left(u_x\right)}{w_n^2}
	%\end{array}
	.
	\label{twostarx}
\end{equation}
\item Multiplying (\ref{vn_t}) and (\ref{vn_x}) and then using (\ref{JCBR}) we can eliminate 
\(w_{n-1}(u_t) w_{n-1}(u_x)\) and \(w_n(u_t) w_n(u_x)\) to obtain the formula relating lattice variables $v_n$ and $w_n$ \cite{ferapont}
	\begin{equation}
		\frac{w_{n-1} w_{n+1}}{w_n^2}=\frac{v_{n, t} v_{n, x}}{\left(v_n-v_{n-1}\right)\left(v_{n+1}-v_n\right)}.
		\label{wnpm1}
	\end{equation}
\end{enumerate}
\end{corollary}
\noi The following theorem was stated without a proof in \cite{demint}.
\begin{theorem*}{Quantity $v_n$ satisfies the lattice equation
(\ref{ShY}) with the left-boundary condition $v_0=0.$}
\end{theorem*}
\begin{proof}
The case $n=1$ can be verified directly if we consider (\ref{v_1}) and the fact that 
\[v_2=\frac{u_{tx}}{uu_{tx}-u_t u_x}.\]
Further, for $n>1,$ on differentiating (\ref{vn_t}) with respect to $x$ and using (\ref{twostarx}), we obtain
	\begin{equation}
		\p^2_{t x}v_{n} =\frac{w_{n+1} w_{n-1}}{w_n^2}\left(\frac{w_n(u_t) w_n(u_x)}{w_n w_{n+1}}
		 -\frac{w_{n-1}\left(u_t\right) w_{n-1}\left(u_x\right)}{w_{n-1} w_n}\right).
		 \label{vn_tx_pre}
	\end{equation}
Then, using formulae (\ref{JCBR})	and (\ref{wnpm1}) we can eliminate variables $w_n$ to obtain (\ref{ShY}).
\end{proof} 

\section{Formulae for integrals and general solutions of the FSY lattice}
Here we discuss the aspects of Darboux integrability of the obtained systems, e.g., the existence of 
characteristic integrals and solutions depending on arbitrary functions. 
Given that both systems (\ref{ShY}) and (\ref{todaln}) have their origin in the scalar equation
(\ref{wn0}) they inherit its Darboux integrability and the integrals. 
Hence the integrals (\ref{defIJ}) will be recalculated in the lattice variables $v_n$ given by formula (\ref{vdef}).
\subsection{Recurrent formula for characteristic integrals}
\subsubsection*{Boundary condition $(0,\infty)$.}
Firstly, consider the boundary condition (\ref{myboundary0N}).
It is convenient to calculate the $t$-integrals using the recurrent formula (\ref{recJpair}). 
To recalculate them in the lattice variables $v_{n}$, 
we re-write (\ref{wnpm1}) and make self-substitution:
\begin{equation}
	\begin{array}{ll}
		\ds \frac{w_{n+1}}{w_n} & \ds =\frac{w_n}{w_{n-1}} \cdot \frac{v_{n, t} v_{n, x}}{\left(v_{n+1}-v_n\right)\left(v_n-v_{n-1}\right)} \\[4mm]
		&\ds =\frac{w_{n-1}}{w_{n-2}} \cdot \frac{v_{n, t} v_{n, x} v_{n-1, t} v_{n-1, x}}{\left(v_{n+1}-v_n\right)\left(v_n-v_{n-1}\right)^2\left(v_{n-1}-v_{n-2}\right)}\\[4mm]
		& \dots \\
		& \ds =\frac{\ds\prod_{m=1}^n v_{m, t} v_{m, x}}{v_1^2\left(v_{n+1}-v_n\right) \ds \prod_{m=1}^{n-1}\left(v_{m+1}-v_m\right)^2},\ \ n\ge 1.
		\end{array}
		\label{selfsub}
\end{equation}
Then, on differentiating and eliminating $v_{m,tx}$ by means of (\ref{ShY}), we obtain
\begin{equation}
	\begin{array}{ll}
	\Lambda_n&\ds =\p_x \ln \frac{w_{n+1}}{w_n} \\
	 &\ds =  -\frac{v_{1, x}}{v_1}+\sum_{m=1}^n\left(\frac{v_{m,xx}}{v_{m, x}}-\frac{v_{m+1, x}}{v_{m+1}-v_m}\right) 
		 +\sum_{m=1}^{n-1} \frac{v_{m, x}}{v_{m+1}-v_m}.
	\end{array}
	\label{lnwndif}
\end{equation}
Thus, formula (\ref{recJpair}) takes the following form
\begin{equation}
	J_{n ; p}=J_{n-1 ; p-1}-\partial_x J_{n-1 ; p}+\Lambda_{n-2} J_{n-1 ; p},
	\label{intrecJ}
\end{equation}
where $\Lambda_{n}$ is given by formula (\ref{lnwndif}).
\vspace{2mm}

\begin{Example}{1.2}
Consider equation (\ref{wn3}) in variables (\ref{vdef}) which implies the condition
\begin{equation}
	v_0=0,\, v_{3}=\infty.
	\label{myboundary}
\end{equation}
Condition (\ref{myboundary}) turns lattice (\ref{ShY}) into the following two-component system
\begin{equation}
	\begin{array}{ll}
		v_{1, tx} &\ds = v_{1, t} v_{1, x} \left(\frac{1}{v_1} - \frac{1}{v_2 - v_1}\right), \\[2mm]
		v_{2, tx} &\ds = \frac{v_{2, t} v_{2, x}}{v_2 - v_1}.
		\end{array}
	\label{ShY3}
\end{equation}
The most straightforward way to calculate the integrals of (\ref{ShY3}) 
is to use formula (\ref{J31J32}), in which we must substitute $u=1/v_1$ 
and eliminate the mixed derivatives by means of (\ref{ShY3}). 
However, the integrals can be calculated more efficiently by means of formula (\ref{intrecJ}).
Recall that $J_{n,0}=0$ and $J_{n,n}=1$ for any $n$.
Then, from formulae (\ref{lnwndif}) and (\ref{intrecJ}) we first obtain the auxiliary expression
\begin{equation}
	J_{2,1}=\Lambda_0=-\frac{v_{1,x}}{v_1}.
\end{equation}
Further, we have 
\begin{equation}
	\Lambda_1=\frac{v_{1,xx}}{v_{1,x}}-\frac{v_{1,x}}{v_1}-\frac{v_{2, x}}{v_2-v_1}
\end{equation}
hence the first integral is given by
\begin{equation}
	\begin{array}{ll}
J_{3,1}	&= -\partial_x J_{2,1}+\Lambda_1 J_{2,1} \\
		&\ds =  \frac{v_{1, x x}}{v_1}-\frac{v_{1, x}^2}{v_1^2} 
		 -\frac{v_{1, x}}{v_1}\left(\frac{v_{1, x x}}{v_{1, x}}-\frac{v_{1, x}}{v_1}-\frac{v_{2, x}}{v_2-v_1}\right) \\[4mm]
		& \ds =  \frac{v_{1, x} v_{2, x}}{v_1\left(v_2-v_1\right)}.
		\end{array}
	\label{w3int}
\end{equation}
Lastly, the integral $J_{3,2}$ reads
\begin{equation}
	\begin{array}{ll}
	J_{3,2}
	&= J_{2,1}  + \Lambda_1 J_{2,2} \\[2mm]
	&\ds = \frac{v_{1, xx}}{v_{1, x}} - \frac{2 v_{1, x}}{v_1} - \frac{v_{2, x}}{v_2 - v_1}.
		\end{array}
	\label{w3int2}
\end{equation}
Thus, the integrals $J_{3,1}$ and $J_{3,2}$ constitute a complete set of integrals for the system (\ref{ShY3}).
\end{Example}
\subsubsection*{Boundary condition $(a,b).$}
Let us quickly discuss how integrals can be re-calculated in the case of the boundary condition (\ref{habibconstr}).
Obviously, one can simply make the inverse transformation
\begin{equation}
	v_n=\frac{\vphi_n-a}{\vphi_n-b},
\label{mobiusshort}
\end{equation}
where the non-essential constant factor $a-b$ was suppressed.

However, instead of recalculating the integrals we would rather recalculate formula (\ref{lnwndif}) in variables $\vphi_n$.
Calculations, similar to (\ref{selfsub}) and (\ref{lnwndif}), yield the formula
\begin{equation}
	\begin{array}{ll}
		\Lambda_n &\ds =  -\frac{\varphi_{1, x}}{\varphi_1-a} -\frac{\varphi_{n, x}}{\varphi_n-b}+\frac{\varphi_{n+1, x}}{\varphi_{n+1}-b} \\[2mm]
		& \ds +\sum_{m=1}^n\left(\frac{\varphi_{m, x x}}{\varphi_{m, x}}-\frac{\varphi_{m+1, x}}{\varphi_{m+1}-\varphi_m}\right) 
		 +\sum_{m=1}^{n-1} \frac{\varphi_{m, x}}{\varphi_{m+1}-\varphi_m}.
		\end{array}
		\label{Lambdaphi}
\end{equation}

\begin{Example}{1.3}  Consider the system obtained from (\ref{ShY3}) by applying (\ref{mobiusshort})
\begin{equation}
	\begin{array}{l}
		\ds \vphi_{1, tx} = \vphi_{1, t} \vphi_{1, x} \left( \frac{1}{\vphi_1 - a} - \frac{1}{\vphi_2 - \vphi_1} \right), \\[3mm]
		\ds \vphi_{2, tx}  = \vphi_{2, t} \vphi_{2, x} \left( \frac{1}{\vphi_2 - \vphi_1} - \frac{1}{b - \vphi_2} \right).
		\end{array}
		\label{ShY3phi}
\end{equation}
Hence, we assume that
\begin{equation}
	\vphi_0=a,\ \ \vphi_3=b.
\end{equation}
Then, formula (\ref{Lambdaphi}) yields
\[
J_{2,1}=\Lambda_0=\frac{\varphi_{1, x}}{\varphi_1-b}-\frac{\varphi_{1, x}}{\varphi_1-a}.
\]
Further, from the same formula we obtain
\begin{equation}
	%\begin{array}{l}
		\Lambda_1  =\frac{\varphi_{2, x}}{\varphi_2-b}-\frac{\varphi_{1, x}}{\varphi_1-b}-\frac{\varphi_{1, x}}{\varphi_1-a} %\\& 
		+\frac{\varphi_{1, x x}}{\varphi_{1, x}}-\frac{\varphi_{2, x}}{\varphi_2-\varphi_1}.
		%\end{array}
\end{equation}
The integrals are given by
\begin{equation}
	\begin{array}{ll}
J_{3,1}&\ds = -\partial_x J_{2,1}+\Lambda_1 J_{2,1} \\[2mm]
		&\ds =\frac{(a - b)\varphi_{1, x} \varphi_{2, x} }{(\varphi_1 - a)(\varphi_2 - \varphi_1)(\varphi_2 - b)}
		\end{array}
	\label{w3intphi}
\end{equation}
and
\begin{equation}
	\begin{array}{ll}
	J_{3,2}
	&\ds = J_{2,1}  + \Lambda_1 J_{2,2}\\[2mm]
	&\ds = -\frac{2 \varphi_{1, x}}{\varphi_1 - a} + \frac{\varphi_{2, x}}{\varphi_2 - b} - \frac{\varphi_{2, x}}{\varphi_2 - \varphi_1} + \frac{\varphi_{1, xx}}{\varphi_{1, x}}
		.
		\end{array}
	\label{w3int2phi}
\end{equation}
\end{Example}
\begin{remark}
	System (\ref{ShY3}), given in the variables 
	\[v_1=-a/p,\ \  v_2=q\] 
	was shown to be Darboux integrable in \cite{demstarts}, in which its integrals and higher symmetries were described. 
	Later, the point-equivalent system (\ref{ShY3phi}) was used as an illustrating example in several papers (see e.g. \cite{habibsak,garhabib}).
\end{remark}
\subsection{General solutions}
The general solution of (\ref{wn0}) has a particularly simple form \cite{darb}, namely
\begin{equation}
	u(t,x)=\sum_{s=1}^{N-1}X_s(x)T_s(t),
	\label{gensolu}
\end{equation}
where $T_s(t)$ and $X_s(x)$ are arbitrary functions of their arguments. Thus, the number of arbitrary functions in it, $2N-2$,
matches the order of equation (\ref{wn0}).

Substituting (\ref{gensolu}) into (\ref{vdef}) we obtain the general solution of the lattice (\ref{ShY})
with the boundary condition (\ref{myboundary0N}):
\begin{equation}
	v_n=\frac{\mbox{det}\left(a_{ij}\right)}{\mbox{det}\left(b_{kl}\right)},
	\label{solnondegen}
\end{equation}
where the components of the determinants are given by
\begin{equation}
	\begin{array}{ll}
	a_{ij} & \ds =\sum_{s=1}^{N-1} X_s^{(i)} T_s^{(j)},\ \ b_{kl}= \sum_{s=1}^{N-1} X_s^{(k-1)} T_s^{(l-1)},\\
	& i,j =1,\dots,n-1,\ \ k,l=1,\dots,n.
	\end{array}
	\label{detcomp}
\end{equation}
and the derivatives are denoted as follows:
\begin{equation}
	%\begin{array}
	X^{(\kappa)}_s=\frac{d^\kappa X}{dx^\kappa},\ \ T^{(\mu)}_s=\frac{d^\mu T}{dt^\mu}.
%\end{array}
\end{equation}
Hence, formula (\ref{solnondegen}) provides us with the general solution of (\ref{ShY}) satisfying
condition (\ref{myboundary0N}).
In order to obtain the general solution satisfying the boundary condition (\ref{habibconstr}) one has to apply
transformation (\ref{mobiusab}). 

Let us consider the degenerate boundary condition
\begin{equation}
	v_0=v_N=a
	\label{degenboundary}
\end{equation}
in more detail. 
Although calculation of integrals in this case is straightforward, 
the general solution requires a more detailed analysis.

Firstly, we consider the homogeneous condition 
\begin{equation}
	v_0=v_{N}=0.
	\label{homogboundary}
\end{equation}
Comparing it with (\ref{vdef}) we see that 
\[w_{N-1}(u_{tx})=0.\] 
The general solution of the latter equation is given by
\begin{equation}
	u=X_0(x)+T_0(t)+\sum_{s=1}^{N-2} X_s T_s.
	\label{homogsol}
\end{equation}
Substituting (\ref{homogsol}) in (\ref{vdef}) and shifting the result, we get the general solution
of (\ref{ShY}) satisfying condition (\ref{degenboundary}) in the form
\begin{equation}
	v_n=\frac{\mbox{det}\left(\alpha_{ij}\right)}{\mbox{det}\left(\beta_{kl}\right)}+a,
	\label{soldegen}
\end{equation}
where
\begin{equation}
	\begin{array}{ll}
	\alpha_{ij} & \ds =\sum_{s=1}^{N-2} X_s^{(i)} T_s^{(j)},\\ 
	\beta_{kl} &\ds = \delta_{k1}T_0^{(l-1)}+\delta_{1l}X_0^{(k-1)}+\sum_{s=1}^{N-2} X_s^{(k-1)} T_s^{(l-1)},\\
	& i,j =1,\dots,n-1,\ \ k,l=1,\dots,n.
	\end{array}
	\label{detcompdegen}
\end{equation}
Here, $\delta_{kl}$ is the Kronecker delta:
\[
\delta_{k l}= \left\{
	\begin{array}{ll}
	1,&  k=l \\ 
	0,& k \neq l
	\end{array}
\right.
\]
\begin{Example}{1.4}  For system (\ref{ShY3}), formula (\ref{solnondegen}) gives a general solution in the following form
\begin{equation}
	\begin{array}{ll}
	v_1 & \ds =\frac{1}{T_1 X_1 + T_2 X_2},\\ 
	v_2 & \ds =\frac{T_1' X_1' + T_2' X_2'}{T_1 X_1 T_2' X_2' - T_1 X_2 T_2' X_1' - T_2 X_1 T_1' X_2' + T_2 X_2 T_1' X_1'}.
	\end{array}
\end{equation}
Applying (\ref{mobiusab}), we get a general solution of (\ref{ShY3phi}) in the form
\begin{equation}
	\begin{array}{ll}
	\vphi_1 &\ds  =a + \frac{a - b}{(a - b) (X_1 T_1 + X_2  T_2) - 1  },\\[3mm]
	 \vphi_2 &\ds =a+\frac{(a - b) \big(T_1' X_1' + T_2' X_2'\big)}{(a - b) \big(X_2 X_1'-X_1 X_2'\big) \big(T_2 T_1'-T_1 T_2'\big) - T_1' X_1' - T_2' X_2'}.
	\end{array}
	\end{equation}
Finally, the solution to the system (\ref{ShY3phi}) in the degenerate case $b=a$ is obtained from the formula (\ref{soldegen})
\begin{equation}
	\begin{array}{ll}
	\vphi_1&\ds =a+\frac{1}{T_1 X_1 + T_0 + X_0},\\ 
	\vphi_2&\ds =a+\frac{T_1' X_1'}{T_0 T_1' X_1' - T_1 T_0' X_1' + X_0 T_1' X_1' - X_1 T_1' X_0' - T_0' X_0'}.
\end{array}
\end{equation}
\end{Example}
\section{Conclusion}
An evident extension of the results presented in this paper concerns the discrete and 
semi-discrete versions of the FSY lattice (\ref{ShY}). 
To derive these, one can start with the discrete scalar equation (\ref{wn0}), 
where derivatives are replaced by shifts of discrete variables, 
and introduce new lattice variables in a similar manner, specifically using formula (\ref{vdef}). 
The subsequent calculations rely on the Jacobi identity (\ref{sylvester}) and largely mirror those presented in this paper.

The resulting equations should exhibit analogous properties. In particular, when the boundary condition (\ref{myboundary0N}) is imposed, they become Darboux integrable. Detailed aspects of integrability, such as characteristic integrals and general solutions, will be addressed in a separate publication.

Another direction, concerning finite reductions of the FSY lattice, involves finding exact solutions to integrable models with two spatial dimensions \cite{habibkhakim,garhabib}. We anticipate that the formulae for the integrals and general solutions provided in this paper will facilitate progress in addressing these problems.
\end{document}